%% file: main.tex
\documentclass[11pt]{article}
\usepackage{graphicx,tikz} 
\input{defs.tex}
\usepackage{enumitem}
\usepackage{array}
\usepackage{multirow}

\input{Qcircuit}

\usepackage[margin=1in]{geometry}

\title{Forrelation is Extremally Hard}
\author{Uma Girish\thanks{  Columbia University. Email:  \href{mailto:ug2150@columbia.edu}{ug2150@columbia.edu}} \and Rocco A.~Servedio\thanks{  Columbia University. Email:  \href{mailto:ras2105@columbia.edu}{ras2105@columbia.edu}}}
\date{}

\begin{document}

\maketitle 

\begin{abstract}

The Forrelation problem is a central problem that demonstrates an exponential separation between quantum and classical capabilities. In this problem, given query access to $n$-bit Boolean functions $f$ and $g$, the goal is to estimate the Forrelation function $\forr(f,g)$, which measures the correlation between $g$ and the Fourier transform of $f$. 

In this work we provide a new linear algebraic perspective on the Forrelation problem, as opposed to prior analytic approaches. 
We establish a connection between the Forrelation problem and \emph{bent Boolean functions} and through this connection,  analyze an \emph{extremal} version of the Forrelation problem where the goal is to distinguish between extremal instances of Forrelation, namely $(f,g)$ with $\forr(f,g)=1$ and $\forr(f,g)=-1$. 

We show that this problem can be solved with \emph{one} quantum query and success probability \emph{one}, yet requires $\tilde{\Omega}\pbra{2^{n/4}}$ classical randomized queries, even for algorithms with a one-third failure probability, highlighting the remarkable power of one exact quantum query. We also study a restricted variant of this problem where the inputs $f,g$ are computable by small classical circuits and show classical hardness under cryptographic assumptions.
\end{abstract}

\input{sections/intro}

\input{sections/hard-distributions}

\input{sections/classical-lower-bound}

\bibliographystyle{alpha}
\bibliography{updated_ref}

 \appendix

\input{sections/appendix_2}

\end{document}

%% file: sections/intro.tex

\section{Introduction}

Understanding the relative power of quantum versus classical computation is one of the major goals in complexity theory. Following the seminal work of Shor~\cite{Sho94}, it is widely believed that quantum computation is exponentially more powerful than classical computation; however, since we are unable to prove classical lower bounds for strong models of computation, there are relatively few settings in which such a separation can be unconditionally established. Query complexity is an important example of such a setting where we can unconditionally prove exponential quantum speedups. 

In query complexity, we typically consider a Boolean function $f:\{0,1\}^n \to \{\pm 1\}$ and the objective is to compute some property of $f$ by querying the values $f(x)$ on as few inputs $x\in \{0,1\}^n$ as possible. In the classical setting, the algorithm can adaptively and probabilistically choose inputs to query, and the goal is to solve the problem with high success probability, say at least $2/3$. In the quantum setting, the standard way to model a quantum query is by means of the unitary operator $O_f$ which maps $\ket{x}$ to $\ket{x}f(x)$ for all $x\in\{0,1\}^n$ and as before, the goal is to compute some property of $f$ with high success probability, while minimizing the number of calls to the unitary $O_f$. Numerous works~\cite{DJ92,Sim97,BV97,Aar10} have demonstrated properties of $f$ that are exponentially easier to compute with quantum queries as opposed to classical queries. For instance, a version of periodicity testing~\cite{cleve04} can be solved with $\mathrm{poly}(n)$ quantum queries, while requiring $2^{\Omega(n)}$ classical randomized queries; this algorithm is a key subroutine in Shor's factoring algorithm~\cite{Sho94}, and the classical query lower bound helps explain some of the difficulty in finding efficient classical algorithms for factoring. 

Understanding the strongest possible separation between quantum and classical computation has long been a topic of great interest. The overarching motivation here is to find a problem that is as easy as possible for quantum algorithms and as hard as possible for classical algorithms. 
In this context, a new problem called the \emph{Forrelation problem} has emerged as a central concept.

In the Forrelation problem, given query access to Boolean functions $f,g$, the goal is to estimate the value of the Forrelation function $\forr(f,g) \in [-1,1]$, which captures the correlation between $g$ and $\widehat{f}$, the Fourier transform of $f$. The Forrelation function has been used to establish numerous results about the power of quantum computation. The first such result is due to Aaronson~\cite{Aar10}, who showed that the Forrelation problem can be solved with high probability using just \emph{one} quantum query, yet randomized algorithms require $2^{\Omega(n)}$ queries. This was quantitatively strengthened by~\cite{AA15} who defined variants of the Forrelation problem that are now known to demonstrate the largest possible separation between quantum and randomized query complexity for partial functions~\cite{AA15,Tal20,sherstov2023optimal,BS21,BGGL22}. The Forrelation problem has since been used to prove many other results, including the celebrated oracle separation of $\mathsf{BQP}$ and $\mathsf{PH}$~\cite{RT22}. Variants of this problem have been used to prove quantum lower bounds, such as a construction of a classical oracle relative to which $\mathsf{P}=\mathsf{NP}$ but $\mathsf{BQP} \neq \mathsf{QCMA}$ ~\cite{AIK22}, as well as the existence of various quantum cryptography primitives~\cite{KQST23,KQT25},  separations between adaptive and non-adaptive quantum algorithms~\cite{GSTW24}, and various separations between quantum and classical communication complexity~\cite{GRT22,GRZ21,AG23}.

However, all these works share one common limitation -- they only establish classical hardness for estimating the Forrelation function \emph{up to a constant less than one}. The best-known result in this context is due to Aaronson and Ambainis~\cite{AA15}, who show that distinguishing between $\forr(f,g)\ge 2/\pi$ and $\forr(f,g)\le -2/\pi$ requires $2^{\Omega(n)}$ classical randomized queries. The factor of $2/\pi$ arises from their analytic approach, which involves sampling Gaussian random variables and rounding them to $\{\pm 1\}$. All existing techniques use this framework and 
run into the same $2/\pi$ barrier. This naturally leads us to ask: just how hard is it to approximate the Forrelation function in the extremal case? In particular, Aaronson and Ambainis~\cite{AA15} ask the following question (see open question \#4 in the discussion section of their paper): if we want a $1$ versus $2^{\Omega(n)}$ separation between quantum and classical query complexity, how small can the error of the quantum algorithm be? More precisely, we ask:
\begin{quote}\centering
    \emph{How hard is it to distinguish $\forr(f,g)=1$ from $\forr(f,g)=-1$?}
\end{quote}
These two extreme cases capture the largest and smallest possible values of the Forrelation function, and hence this question captures the hardness of approximating Forrelation to any non-trivial factor. 

To study this problem, we introduce a fundamentally new way of looking at the Forrelation problem.  In contrast to previous analytic approaches, which rely on rounding high-dimensional Gaussian distributions, 
our approach uses only simple linear algebra over $\F_2^n$ and elementary probabilistic arguments. We establish a novel connection between the Forrelation problem and \emph{bent} functions, a well-studied concept in the analysis of Boolean functions. Using this connection, we show that despite the strong promise on the inputs, the extremal Forrelation problem is classically hard. Our main theorem establishes a bounded-error randomized lower bound of $2^{\Omega(n)}$ for this problem. In contrast, there is a simple quantum algorithm~\cite{Aar10} that solves this problem with one quantum query and success probability one.

In the following section, we formally define the Forrelation problem, describe the history of the problem, and state our main results.

\subsection{Forrelation Problem}
To describe the Forrelation problem, we first need to introduce the concept of the Fourier transform and the Forrelation function. The Boolean Fourier transform, also known as the Walsh-Hadamard transform, is a central concept in Boolean function analysis which has applications to learning theory, social choice theory, circuit complexity, property testing, and quantum versus classical separations. It is defined as follows.  
\begin{definition}[Fourier Transform]\label{def:fourier_transform} 
For $f:\bin^n\to \R$, define $\widehat{f}:\bin^n\to \R$ 
by
\[\widehat{f}(y):= \frac{1}{2^n}\sum_{x\in\bin^n}f(x) (-1)^{\abra{x,y}}\text{ for all }y\in\bin^n. \]
\end{definition}

We now define the Forrelation function, which has close connections with the Fourier transform of Boolean functions. The input to this function consists of the truth tables of two Boolean functions $f$ and $g$ and the output is the correlation between $g$ and the Fourier transform of $f$.

\begin{definition}[Forrelation Function] The function $\forr$ is defined as follows. For Boolean functions $f,g:\bin^n\to \{\pm 1\}$, define
\begin{align}
\forr(f,g)&:=\frac{1}{2^{n/2}}\sum_{y\in \bin^n}\widehat{f}(y)g(y) \label{eq:def_forr}\\ 
& = \frac{1}{ 2^{3n/2}} \sum_{x,y\in\bin^n} f(x)g(y)(-1)^{\abra{ x,y}}. \nonumber
\end{align}
\end{definition}
It is not too difficult to see that for any pair of Boolean functions $(f,g)$, the value of $\forr(f,g)$ is between $-1$ and $1$. Indeed, applying Cauchy-Schwarz on~\Cref{eq:def_forr} implies that $ |\forr(f,g)|\le 2^{-n/2} \sqrt{\sum_y \widehat{f}(y)^2} \sqrt{\sum_y g(y)^2}.$ Since $f$ and $g$ are $\pm 1$-valued functions, Parseval's theorem implies that $\sum_y \widehat{f}(y)^2=\Ex_{\bm{x}}[f(\bm{x})^2]=\Ex_{\bm{y}}[g(\bm{y})^2]=1$ and it follows that $|\forr(f,g)|\le 1$.

\begin{figure}
\centering
\mbox{ 
\Qcircuit @C=1em @R=.7em {
&\ket{+}  &  &  \qw & \multigate{4}{O_{f,g}}  & \ctrl{1} \qwx[1]  & \ctrl{2} \qwx[1]  & \ctrl{3} \qwx[1]& \ctrl{4} \qwx[1] & \gate{H}& \meter \\
& \ket{+} &  &     \qw & \ghost{O_{f,g}}  &\gate{H} & \qw  &\qw &\qw  & \qw  \\ 
& \ket{+}  &  &   \qw & \ghost{O_{f,g}} &\qw &\gate{H} &\qw &\qw  & \qw \\
& \ket{+}  &  &   \qw  & \ghost{O_{f,g}}& \qw &\qw &\gate{H} &  \qw &   \qw   \\
& \ket{+}  & &   \qw & \ghost{O_{f,g}}& \qw & \qw & \qw & \gate{H}  &  \qw   \\
}
}
\caption{Quantum circuit for the Forrelation problem for $n=5$. Here, $\ket{+}=\tfrac{1}{\sqrt{2}}(\ket{0}+\ket{1})$ and $O_{f,g}$ is the oracle mapping basis states $\ket{0}\ket{x}$ to $\ket{0}\ket{x}f(x)$ and $\ket{1}\ket{x}$ to $\ket{1}\ket{x}g(x)$.}
\label{fig:qcircuit}
\end{figure}

The Forrelation function is also interesting from the perspective of quantum algorithms as it can be interpreted as the bias of a certain one-query quantum algorithm.  More precisely, Aaronson~\cite{Aar10} gave a simple quantum query algorithm that makes one call to $O_{f,g}$ and returns $1$ with probability precisely $\frac{1}{2}+\frac{\forr(f,g)}{2}$ (see~\Cref{fig:qcircuit} for an illustration of the algorithm.) The intuition for this quantum algorithm comes from the ability of quantum circuits to implement the Fourier transform.
We can view the Fourier transform as a unitary map that transforms the truth table of $f$ into that of $\widehat{f}$ (up to a normalization factor of $2^{n/2}$). Additionally, this unitary map turns out to be a Hadamard matrix, i.e., a tensor product of $n$ Hadamard gates. In contrast, it seems difficult to estimate the Forrelation of $f$ and $g$ using only classical queries to the truth tables of $f$ and $g$. This motivates the definition of the Forrelation problem, where the goal is to estimate the Forrelation function up to a small additive error. 

\begin{definition}[\textsc{Forrelation Problem}] \label{def:forr_problem}
Fix a parameter $0\le \eps\le 1$. Given query access to the truth tables of Boolean functions $f,g:\bin^n\to \{\pm 1\}$ that are promised to satisfy either
\begin{itemize}
    \item \yes case: $\forr(f,g)\ge \eps$, or
    \item  \no case: $\forr(f,g)\le -\eps$, 
\end{itemize}
distinguish between the two cases.\footnote{In prior works, the goal of the Forrelation problem is usually to distinguish $\forr(f,g)\ge \eps$ from $\forr(f,g)\le \eps/2$.  Nevertheless, existing lower bounds also work for our variant of the problem; the proofs can be modified to produce two distributions supported on $\forr(f,g)\ge \eps$ and $\forr(f,g)\le -\eps$, respectively, such that they are each indistinguishable from the uniform distribution over $(f,g)$ for classical algorithms of small cost.}
\end{definition}

As mentioned before, there is a simple quantum algorithm that solves the Forrelation problem with one query, and the success probability of the algorithm is precisely $\frac{1}{2}+\frac{\eps}{2}$ where $\eps$ is the underlying parameter. In particular, when $\eps=1$, the algorithm makes no error.  

\paragraph*{Limitations of Prior Works on Forrelation.} Numerous works have established classical hardness of the Forrelation problem for $\eps$ bounded away from 1~\cite{Aar10,AA15,RT22,BS21}. We will describe these results in more detail in~\Cref{fig:tab2}, but all these works share a common limitation, which is that they only establish hardness for estimating the Forrelation \emph{up to a global constant less than one}. The best-known constant is due to Aaronson and Ambainis~\cite{AA15}, who showed that the Forrelation problem with  $\eps=2/\pi-o(1)$ requires $2^{\Omega(n)}$ randomized queries. 

\subsection{Our Results}

We consider the Forrelation problem with $\eps=1$, and call this the \maxforrelationproblem. Here, we are given query access to Boolean functions $f,g:\{0,1\}^n\to \{\pm 1\}$ that are promised to satisfy $|\forr(f,g)|=1$, and we wish to tell whether $\forr(f,g)=1$ or $\forr(f,g)=-1$. As mentioned before, the quantum algorithm for this problem follows immediately from the works of~\cite{Aar10,AA15}. When this algorithm is run on inputs to the \maxforrelationproblem ~it makes one query and solves the problem with success probability one.
Our main theorem is that the classical {randomized bounded-error} query complexity of this problem is $2^{\Omega(n)}$. 
\begin{theorem}\label{thm:main_theorem} The~\maxforrelationproblem, which is solvable with one quantum query and success probability one, requires $\tilde{\Omega}(2^{n/4})$ queries for any classical randomized query algorithm that succeeds with at least $2/3$ probability.
\end{theorem}

We also analyze a variant of this problem where the oracles $f,g$ are efficiently computable.


\begin{corollary}\label{thm:main_corollary} Suppose one-way functions exist against classical $\poly(n)$-time algorithms. Then, there is no $\poly(n)$-time randomized algorithm that solves the~\maxforrelationproblem~with probability at least $2/3$, even if the oracles $f,g$ are computable by $\poly(n)$-sized classical circuits.
\end{corollary}

We remark that the above result is proved in the black-box setting where the algorithm can only query the truth tables of $f,g$. We do not know if these results apply in the white-box setting where the algorithm is given an explicit description of a small circuit computing $f,g$. 

\paragraph*{}We now mention a few consequences of our results.
 
\paragraph*{Lower Bounds for the Forrelation Problem.} 

\begin{figure}
\begin{center}
\begin{tabular}{ |m{6em}|m{5em}|m{13em}| m{11em}| } 
 \hline 
 &   Forrelation Problem & Classical Model & Classical Lower Bound  \\  
  \hline
~\cite{Aar10}  & $\eps=0.05$   & Randomized decision tree & $\tilde{\Omega}(2^{n/4})$ queries (depth)
  \\ \hline 
  ~\cite{AA15} & $\eps=2/\pi$  &Randomized decision tree   & $\tilde{\Omega}(2^{n/2})$ queries \\ \hline
    ~\cite{RT22} & $\eps=\Theta(1/n)$  &Randomized AC0 circuits   & $\exp\pbra{2^{\Omega(n/d)}}$ size \\ \hline
    ~\cite{BS21} & $\eps=1/2^{10}$  &Randomized AC0 circuits   & $\exp\pbra{2^{\Omega(n/d)}}$ size \\ \hline
  {\bf  [This work]} & $\eps=1$  &{\bf Randomized decision tree}  & $\tilde{\Omega}(2^{n/4})$ {\bf queries} \\ \hline
 \end{tabular}\caption{Lower Bounds for the Forrelation Problem}\label{fig:tab2}
\end{center}\end{figure}

Numerous works have studied lower bounds for the Forrelation problem. See~\Cref{fig:tab2} for a summary. 
The first classical lower bound for the Forrelation problem was established by Aaronson~\cite{Aar10}, who showed a $\tilde\Omega(2^{n/4})$ lower bound when $\eps=0.05$. This was improved to a $\tilde{\Omega}(2^{n/2})$ lower bound for $\eps=2/\pi-o(1)$  by Aaronson and Ambainis~\cite{AA15}. In their breakthrough result, Raz and Tal~\cite{RT22} proved lower bounds when $\eps=\Theta(1/n)$. Although this choice of $\eps$ is smaller than in the previously mentioned works, the true strength of~\cite{RT22} lies in their classical lower bound which holds against a much more powerful model than classical query algorithms. Bansal and Sinha~\cite{BS21} strengthened this result by proving it in the regime of $\eps=\Theta(1)$. It is worth emphasizing that a crucial aspect of this work -- and a key reason they were able to resolve the conjecture of Aaronson and Ambainis on maximal separations -- was their focus on lower bounds in the regime of $\eps=\Theta(1)$ as opposed to $\eps=\Theta(1/n)$. All of this underscores the difficulty and importance of proving lower bounds for the Forrelation problem, particularly as $\eps$ increases. In particular, as mentioned earlier,~\cite{AA15} have asked the following question (open question \#4): how large $\eps$ can be for the Forrelation problem while remaining classical hard with $2^{\Omega(n)}$ queries? Our main theorem shows that even when $\eps$ is as large as can be (i.e., one), the Forrelation problem requires $2^{\Omega(n)}$ classical queries.   



\paragraph*{Efficient Oracle Separation.} In~\cite{AC17}, Aaronson and Chen ask, what happens if we consider quantum algorithms that can access an oracle, but we impose a constraint that the oracle has to be “physically realistic”? The motivation is to design a quantum advantage experiment by studying query complexity separations where the input oracles are implementable by small classical circuits. Motivated by this, we ask
\begin{quote}\emph{How hard is it to estimate $\forr(f,g)$ when $f,g$ are computable by $\poly(n)$-sized circuits?}\end{quote}
As an easy corollary of our main result (\Cref{thm:main_corollary}), we show that if (classically secure) one-way functions exist, then there is no classical $\poly(n)$-time algorithm for the Forrelation problem, even if $f,g$ are computable by polynomial-sized circuits. As mentioned before, our result is a query complexity lower bound that holds in the black-box setting where the algorithm can only query the truth tables of $f,g$ and does not have an explicit description of the circuits computing $f,g$.

\begin{figure}
\begin{center}
\begin{tabular}{ |m{11em}|m{8em}|m{5em}|m{5em}| } 
 \hline
 &  Exact Quantum  & \multicolumn{2}{|c|}{Classical Lower Bound}  \\ 
  \cline{3-4}
& Queries & Queries & Success Probability\\ 
  \hline
 Deutsch-Jozsa~\cite{DJ92}  & 1 & $2^{\Omega(n)}$ & 1 \\  
 &   & O(1) & $2/3$ \\ \hline
 Simon's Problem~\cite{Sim97,BH97} & $O(n)$ (adaptive)   & $2^{\Omega(n)}$ & $2/3$ \\  \hline 
 Welded Tree~\cite{childs03,LLL23}  & $O(n^{2.5})$  (adaptive)    & $2^{\Omega(n)}$ & $2/3$ \\  \hline
Order Finding (QFT) ~\cite{MZ04,cleve04}  & $O(n)$ (parallel), can be made $1$      & $2^{\Omega(n)}$ & $2/3$ \\  \hline
Hidden Linear Structure (QFT)~\cite{BCW02}  & $O(n)$  (parallel), can be made $1$      & $2^{\Omega(n)}$ & $2/3$ \\  \hline
{\bf \maxforrelationproblem~ [this work]} & 1   &  $2^{\Omega(n)}$ & $2/3$ \\ \hline
 \end{tabular}\caption{Speedups with Exact Quantum Computation}\label{fig:tab1}\end{center}
\end{figure}

\paragraph*{Power of Exact Quantum Computation.}

Exact algorithms are  algorithms that make no error, that is, they succeed with probability one. 
The power of exact quantum algorithms has been studied extensively in the past; see~\Cref{fig:tab1} for a summary. 
Numerous works have tried to understand the best possible separations between exact quantum and classical algorithms for total functions~\cite{Amb12,ABB17,BCWZ99}. For partial functions, one of the earliest results is due to Deutsch–Jozsa~\cite{DJ92}, who showed that distinguishing between the constant function and a balanced function can be solved by an exact quantum algorithm with one query, but any zero-error randomized algorithm requires $2^{\Omega(n)}$ queries. However, the randomized query complexity drops to $O(1)$ if the algorithm is allowed to err with small probability. The first exponential speedup over bounded-error randomized algorithms is demonstrated by Simon's problem~\cite{Sim97}, which can be solved exactly~\cite{BH97} with $O(n)$ quantum queries, but requires $2^{\Omega(n)}$ classical queries in the bounded-error model. Since then, there have been other problems showing a similar separation using Discrete Fourier transforms, with the additional advantage that all the queries can be made in parallel~\cite{cleve04,MZ04,BCW02}.  While many of these works describe their quantum algorithms as making one query, in their models one can retrieve the truth table values at $n$ different points with just one query; this corresponds to making $O(n)$ parallel queries in the standard model. In particular, they consider query access of the form $\ket{x}\to\ket{f(x)}$ where $f$ has $n$-bit outputs, whereas the standard model only allows single-bit outputs. 
At first glance, it appears that such quantum algorithms need to make $\Omega(n)$ queries. However, using this, one can obtain a new Boolean function on $2n$-bit inputs with an \emph{exact one-query} quantum algorithm such that every randomized algorithm requires $2^{\Omega(n)}$ queries~\footnote{We thank the anonymous reviewer for TQC 2025 pointing this out.\label{footnote:2}}.  The idea here is to replace  $f(x)$ by the Hadamard encoding of $f(x)$ and to use the Bernstein-Vazirani algorithm. Our main result (\Cref{thm:main_theorem}) achieves a similar separation, arguably for a simpler function and a simpler quantum protocol.

\subsection{Outlook \& Future Directions}

We now highlight a number of open questions inspired by our work.

\paragraph*{Optimal Separations.} The bounded-error randomized query complexity of~\maxforrelationproblem~remains to be understood. The results of~\cite{AA15,BGGL22} implies a $\tilde{O}(2^{n/2})$ upper bound, while we prove an $\tilde{\Omega}(2^{n/4})$ lower bound in~\Cref{thm:main_theorem}. We conjecture that our lower bound can be improved to $\tilde{\Omega}(2^{n/2})$, matching the upper bound from~\cite{AA15,BGGL22}. This would recover the separation from~\cite{AA15}, with the additional advantage that the quantum algorithm succeeds with probability one. 
\begin{conjecture}
The bounded-error randomized complexity of the \maxforrelationproblem~is $\tilde{\Omega}(2^{n/2})$.
\end{conjecture} 
More broadly, one can ask about optimal separations between quantum and classical query complexity. This question was studied by~\cite{AA15}, who introduced a variant of Forrelation called $k$-Forrelation whose quantum query complexity is $\lceil k/2\rceil$ and conjectured that its randomized query complexity is $\tilde{\Omega}(2^{n\cdot (1-1/k)})$. When $k=2$, this problem is identical to the Forrelation problem as in~\Cref{def:forr_problem}. They further conjectured that this is the best possible separation between bounded-error quantum and randomized query complexity. There have been a number of works on this topic~\cite{BGGL22,Tal20}, culminating in the works of~\cite{sherstov2023optimal} and~\cite{BS21} that proved this conjecture. One can ask if a similar separation can be achieved by \emph{exact} quantum algorithms. We conjecture that this can indeed be achieved.
\begin{conjecture}
For $k>2$, the bounded-error randomized complexity of the extremal version of the $k$-Forrelation problem is $\tilde{\Omega}(2^{n (1-1/k)})$.
\end{conjecture}

Resolving this conjecture would shed light on the best possible separations between exact quantum and bounded-error randomized query complexity for partial functions. 

\paragraph*{Forrelation of Low-Degree Polynomials}
Recently, there has been a lot of interest in studying the complexity of forrelation for instances in which $f$ and $g$ are degree-$d$ $\F_2$-polynomials~\cite{Geo25,Shu25}. It is not too difficult to define pairs of low-degree polynomials that have large Forrelation, for instance, in~\Cref{def:hard_distributions}, we can sample $h$ to be a random degree-$d$ $\F_2$-polynomial as opposed to a uniformly random Boolean function. This variant of Forrelation has been studied with the broad motivation of finding an efficient instantiation of an oracle which makes Forrelation classically hard. Towards this, researchers have tried to understand the  computational complexity of estimating Forrelation for low-degree polynomials~\cite{Geo25,Shu25}.

We believe that understanding the query complexity of this version of Forrelation is independently interesting. In this setting, given query access to low-degree polynomials $f,g$, we wish to compute $\forr(f,g)$ using as few queries to the truth tables of $f$ and $g$ as possible. One trivial algorithm for this is to learn the polynomials using $\binom{n}{\le d}$ queries each and then compute Forrelation offline. We conjecture that this algorithm is more or less optimal.
\begin{conjecture}
The \maxforrelationproblem~where the inputs are restricted to degree-$d$ $\F_2$-polynomials has randomized query complexity $n^{\Omega(d)}$. 
\end{conjecture}
Proving this conjecture would give some evidence towards the hardness of computing forrelation even when the inputs are low-degree polynomials -- it would indicate that computing the forrelation of low-degree polynomials is almost as hard as learning the polynomials. 

\paragraph*{Hardness Against Stronger Classical Models}

One of the original motivations in~\cite{Aar10}  for introducing the Forrelation problem is that it was conjectured to be classically hard to compute, even for the relatively powerful model of $\mathsf{AC}^0$ circuits. This conjecture was proved in the breakthrough result of Raz and Tal~\cite{RT22}, resulting in an oracle separation between $\mathsf{BQP}$ and $\mathsf{PH}$. We conjecture that a variant of the extremal Forrelation problem can be used to demonstrate a similar separation. 
\begin{conjecture}
Depth-$d$ $\mathsf{AC}^0$ circuits require $\exp\pbra{2^{\Omega(n/d)}}$ size to solve the~\maxforrelationproblem. 
\end{conjecture}
We suspect that proving this conjecture would require studying a richer class of bent functions. We think this question is also interesting from the perspective of classical complexity theory, as it may lead to new techniques to prove $\mathsf{AC}^0$ lower bounds.

\subsection{Our Ideas}
Before we go into our proof ideas, we first describe the existing approaches to proving Forrelation lower bounds and explain why they fail to work as $\eps$ tends to $1$. 

\subsubsection{\texorpdfstring{Why prior approaches don't work for $\eps=1$.}{Why prior approaches don't work for eps=1.}} To prove randomized lower bounds on the Forrelation problem, one needs to produce hard distributions, i.e., distributions on the \yes and \no instances of the problem that are classically hard to distinguish. For now, let us focus on generating \yes instances. It is not too difficult to generate a pair of \emph{real-valued} functions $f,g$ with large Forrelation and $\E[f^2]=\E[g^2]=1$; indeed, given any nonzero function $f$, we may define $g$ to be proportional to $\widehat{f}$, with an appropriate normalization factor so as to satisfy $\forr(f,g)=1$ and $\E[f^2]=\E[g^2]=1$. The difficulty is in producing a pair of $\pm 1$-valued functions with large Forrelation. Existing techniques to generate such pairs of functions follow essentially the same framework: 

\begin{enumerate}

 \item\label{gen:item1} For each $x\in \bin^n$, independently sample $\boldf(x)$ according to the Gaussian distribution with mean 0 and variance $\Theta(\eps)$. For each $y\in \bin^n$, define $\bg(y):=2^{n/2}\widehat{\boldf}(y)$. 
 
 \item\label{gen:item2}  For each $x,y\in\bin^n$, independently round each $\boldf(x),\bg(y)$ to $\pm 1$.  
 
 \end{enumerate} 
 
 It is not difficult to show that~\Cref{gen:item1} generates a distribution on pairs of real-valued functions that with high probability have large Forrelation, namely at least $\Theta(\eps)$. The goal of~\Cref{gen:item2} is to modify these functions to produce $\pm 1$-valued functions that continue to have large Forrelation, at least $\Theta(\eps)$. This is where the choice of the rounding function and the parameter $\eps$ become crucial. The best parameters to date are obtained in \cite{AA15}, where \Cref{gen:item1} is as described above with $\eps=1$ 
 and the rounding in \Cref{gen:item2} is done using the sign function, which maps non-negative reals to $1$ and negative reals to $-1$. Using properties of the Gaussian distribution,~\cite{AA15} show that the resulting Boolean functions (namely $\sign(\boldf),\sign(\bg)$) have Forrelation at least $2/\pi$ in expectation.\footnote{An intuition for the $2/\pi$ factor is as follows. A basic fact about the Gaussian distribution is that for nearly-uncorrelated Gaussians $\boldf(x)$ and $\bg(y)$, we have $\E[\sign(\boldf(x))\cdot \sign(\bg(y))]\approx \frac{2}{\pi}\E[\boldf(x)\cdot \bg(y)]$. Since the Forrelation function is a linear combination of terms of the form $\boldf(x)\cdot \bg(y)$, we get that $\E[\forr(\sign(\boldf),\sign(\bg))]\approx 2/\pi$.} 

One might wonder if there is a different rounding procedure that produces a distribution on inputs with Forrelation very close to 1, but we suspect that this is not the case. The reason is that this framework is actually fairly oblivious to the underlying unitary matrix. In particular, one can consider the problem of estimating the \emph{Rorrelation} function, 
\[
\rorr_U(f,g):=2^{-n} \sum_{x,y\in\bin^n}f(x)g(y)U_{x,y},
\] where $U$ is any $2^n\times 2^n$ unitary matrix; note that when $U=H^{\otimes n}$ this is identical to the Forrelation function. The Rorrelation problem was introduced and studied by~\cite{Tal20}. It turns out that the problem of estimating $\rorr_U(f,g)$ is classically hard for any unitary matrix whose entries are small (at most $1/2^{\Omega(n)}$ in magnitude), furthermore, this can be proved by using the above framework. In particular, if we sample a Haar random orthogonal matrix $\bU$, all of its entries will be small in magnitude and the framework described above will establish classical hardness of estimating $\rorr_{\bU}(f,g)$ up to a small global constant. On the other hand, the extremal version of the Rorrelation problem is vacuously easy for a Haar random orthogonal matrix $\bU$.
In more detail, in \Cref{lem:max_rorr} we prove the following:
\begin{claim} \label{claim:rorr}
With extremely high probability over a Haar random orthogonal matrix $\bU$, 
\begin{equation} \label{eq:max}\max_{f,g:\bin^n\to \{\pm 1\}}|\rorr_{\bU}(f,g)| < 0.99.\end{equation}
\end{claim}
In other words, for a typical Haar random orthogonal matrix $\bU$, there are \emph{no} $\pm 1$-valued functions for which $\rorr_{\bU}(f,g)=\pm 1$ and thus, the extremal version of the Rorrelation problem becomes vacuously easy. 
There is something special about the Forrelation problem and the Hadamard matrix which makes it possible for there to exist even one pair of $\pm 1$-valued functions $(f,g)$ such that $|\forr(f,g)|=1$. 

It turns out that these extremal instances of the Forrelation problem arise from an important class of Boolean functions known as \emph{bent} functions; we now describe this connection.

\subsubsection{Connections to Bent Functions}
\label{sec:bent_functions}

Let us try to generate Boolean functions $f,g:\bin^n\to \{\pm 1\}$ with as large Forrelation value as possible. To do this, let us first revisit the argument for why $\forr(f,g)$ cannot be larger than $1$. Recall that 
\[\forr(f,g)\triangleq 2^{-n/2}\sum_{y\in \bin^n}\widehat{f}(y)g(y)  \le 2^{-n/2}\cdot \sqrt{\sum_y \widehat{f}(y)^2}\cdot \sqrt{\sum_y  g(y)^2}=1,
\]
where the inequality is by Cauchy-Schwarz. For this to be tight, we need to set each $g(y)$ to be a multiple of $\widehat{f}(y)$. Due to the normalization factor, it turns out that we need to set $g(y)=2^{n/2}\cdot \widehat{f}(y)$. In particular, since $g(y)$ is $\pm 1$, it implies that each Fourier coefficient $\widehat{f}(y)$ is $\pm 2^{-n/2}$. Boolean functions with this property, namely that all of the Fourier coefficients have equal magnitude, are known as \emph{bent} Boolean functions. These are precisely the functions that give rise to extremal instances of the Forrelation problem. In more detail, as shown above any pair of functions $(f,g)$ whose Forrelation value is 1 must arise from a bent function $f$; conversely, it is easy to see that any bent function $f$ gives rise to a Boolean function $g$ (namely $g(y):=2^{n/2}\cdot \widehat{f}(y)$) such that $(f,g)$ has Forrelation 1. 

Bent functions are extensively studied in Boolean function analysis. 
One important class of bent functions is the Maiorana-McFarland class of bent functions (see e.g.~Section~6.1 of \cite{CM16}). This family consists of all $n$-bit Boolean functions $(-1)^f$ where $f:\bin^n\to \bin$ is defined at $x\in \bin^n$ by $f(x):=\abra{ x_\up,x_\dn} + h(x_\dn)$ where $n$ is even, $x_\up$ and $x_\dn$ denote the first and second halves of $x$, and $h:\bin^{n/2}\to \bin$ is any Boolean function. Here, $\langle x_\up, x_\dn\rangle$ is the inner product function (mod 2). It is not too difficult to show that every function of this form is bent (see \cite{Dillon74}; a proof of this is implicit in the proof of our \Cref{lemma:hard_distributions}); intuitively, the inner product function itself is a bent function and adding any Boolean function that only depends on the second half of $x$ preserves the bent-ness. We will use the Maiorana-McFarland family to construct hard instances for our lower bound.

\paragraph*{Related Works.} There have been a few works~\cite{Martin2010,CKOR13,ABOS25} that observe the close connections between bent functions and the hidden shift problem (Simon's problem). Independently of our work,~\cite{Suman24} observed the connection between bent functions and extremal instances of the Forrelation problem~\footnote{We thank anonymous reviewers for TQC 2025 for pointing this out.}.

A related problem is property testing for the class of bent Boolean functions. Here, we are given query access to a Boolean function $f$ and the objective is to tell whether $f$ is bent or $\Omega(1)$-far from every bent function. This problem was studied by~\cite{BSK14,GWX13}. In particular,~\cite{GWX13} established an $\Omega(2^{n/4})$ lower bound on the bounded-error randomized query complexity of this problem and they used the Maiorana-McFarland family to construct hard distributions. In our problem, we are given truth table access to $f$ and $g$ that are promised to be bent and that satisfy either $g= \widehat{f}\cdot 2^{n/2}$ or $g=- \widehat{f}\cdot 2^{n/2}$ and we wish to tell these apart. As both $f$ and $g$ are promised to be bent and there is a high degree of correlation between $f$ and $g$, this significantly complicates the problem and we need to introduce some new ideas in our proof.

\subsubsection{Proof Sketch}

The high-level structure of our argument is reminiscent of the lower bound argument of~\cite{GWX13} (see Section 4 of~\cite{GWX13}). For ease of notation, let $\chi:\bin\to \{\pm 1\}$ be the function mapping $x\to (-1)^x$. For a Boolean-valued function $f:\bin^n\to \bin$,  let $\chi_f$ denote the function $(-1)^f$. 

To explain the ideas underlying our proof, we instead consider the problem of distinguishing $\forr(f,g)=1$ from $\forr(f,g)=1/\sqrt{2^n}$. It is simpler to describe a pair of hard distributions for this variant of the problem, so we focus on it in the current sketch; however, the main section of our paper directly tackles the case of $\forr(f,g)=1$ versus $\forr(f,g)=-1$. For the distributions we work with in the proof overview (\Cref{eq:proof_1_no,eq:proof_1_yes}), it is easier to pinpoint exactly where the classical hardness arises from, whereas the actual distributions we work with (\Cref{def:hard_distributions}) have more complicated terms. 

We will now define a distribution over pairs of functions $(\chi_f,\chi_g)$ where $f,g:\bin^n\to \bin$ are Boolean functions. For the \yes distribution, we will sample $\chi_{\boldf}$ to be a bent function.  As described before, this determines a unique $\pm 1$-valued function $\chi_{\bg}$ (namely, $\chi_{\bg}:=2^{n/2}\widehat{\chi_{\boldf}}$) such that $\forr(\chi_{\boldf},\chi_{\bg})=1$.\footnote{Throughout the paper we use {\bf bold font} to indicate random variables; note that in the current context both $\boldf$ and $\bg$ are random variables, albeit (completely) correlated ones.} The process to generate $\boldf$ consists of two steps:
\begin{enumerate}
    \item\label{step1} \textbf{Maiorana-McFarland family:} First, sample a random bent function from the Maiorana-McFarland family. We use this step to ensure that we only produce bent functions -- this is essential to construct instances with $\forr(\chi_f,\chi_g)=1$. 
    \item\label{step2} \textbf{Linear Transformation:} We then apply an invertible linear transformation on the input variables; doing this simply permutes the Fourier coefficients, and hence preserves bent-ness. We will show this step sufficiently masks the input and ``hides'' the inner-product structure, which effectively prevents classical algorithms from detecting structural properties by reading just a few coordinates. 
\end{enumerate}  
In contrast, for the \no distribution, we will sample both $f,g$ to be highly non-bent functions that mimic the Maiorana-McFarland family of bent functions. We now give a more formal description of this process. Sample $\bA\in\bin^{n\times n}$ to be a random invertible matrix, $\bh:\bin^{n/2}\to \bin$ to be a uniformly random function, and let $\bB=(\bA^T)^{-1}$. For the \yes distribution, define 
\begin{equation} \label{eq:proof_1_yes} \boldf(x):=\abra{\bA_\up x,\bA_\dn x}+\bh(\bA_\dn x)\quad \text{and } \quad \bg(y):=\abra{\bB_\up y,\bB_\dn y}+\bh(\bB_\up y) ,\end{equation}
where $M_\up,M_\dn$ denote the upper and lower halves of a matrix $M$. For the \no distribution,
\begin{equation} \label{eq:proof_1_no} \boldf(x):=\bh(\bA_\dn x)\quad \text{and } \quad \bg(y):=\bh(\bB_\up y) ,\end{equation}
We then let $(\chi_{\boldf},\chi_{\bg})$ be the inputs to the Forrelation problem. It is not too difficult to show that $\boldf,\bg$ in~\Cref{eq:proof_1_yes} satisfy $\forr(\chi_{\boldf},\chi_{\bg})=1$ and similarly, in~\Cref{eq:proof_1_no} satisfy $\forr(\chi_{\boldf},\chi_{\bg})=1/\sqrt{2^n}$. The proof of this involves analyzing the Fourier transform of the Maiorana-McFarland family and applying a suitable change of variables; a more general version of this statement is proved in~\Cref{lemma:hard_distributions}.

\paragraph*{Classical Hardness.} We will now give some intuition as to why these distributions are hard to distinguish by classical algorithms. Consider a randomized decision tree of depth $\ell$ that distinguishes these distributions with large advantage. By Yao's minmax principle, fixing the randomness of this algorithm, there exists a deterministic decision tree with the same properties. For now, let us assume that the tree is non-adaptive, that is, it fixes a set of $\ell$ points and this same set is queried along every root-to-leaf path. (This assumption turns out to be not so important, and the analysis for general adaptive algorithms closely follows the non-adaptive case.)
Let $x_1,\ldots,x_k\in\bin^n$ and $y_1,\ldots,y_{\ell-k}\in\bin^n$ be the points queried by the non-adaptive decision tree. Here, $x_i$ and $y_j$ indicate points in the truth table of $\boldf$ and $\bg$ respectively and $k\in \{0,\ldots,\ell\}$; in other words, the algorithm just queries $\boldf(x_1),\ldots,\boldf(x_k)$ and $\bg(y_1),\ldots, \bg(y_{\ell-k})$. We can assume without loss of generality that the $x_i$ are distinct and similarly the $y_j$ are distinct (there may be collisions between $x_i$ and $y_j$). We will assume for the proof sketch that none of $x_i,y_j$ are zero, since such queries are useless in distinguishing the \yes and \no distributions (as the outcomes are the same for both distributions).

Let us look at the contribution of $\bh(\circ)$ to each of $\boldf(x_1),\ldots,\boldf(x_k)$ and $\bg(y_1),\ldots, \bg(y_{\ell-k})$. Recalling the definition of the \yes and \no distribution in~\Cref{eq:proof_1_yes,eq:proof_1_no}, it is easy to see that for both these distributions, the contribution of $\bh(\circ)$ is given by
\begin{equation}\label{eq:sequence_h} \bh(\bA_\dn x_1),\ldots,\bh(\bA_\dn x_k)\quad \text{and}\quad \bh(\bB_\up y_1),\ldots,\bh(\bB_\up y_{\ell-k}
).\end{equation}
In particular, the sequence of points on which $\bh(\circ)$ is implicitly queried is 
\begin{equation}\label{eq:sequence} \bA_\dn x_1,\ldots,\bA_\dn x_k\quad \text{and}\quad \bB_\up y_1,\ldots,\bB_\up y_{\ell-k}.\end{equation}
The main technical lemma of our work shows that if $\ell\le 2^{cn}$, for a suitable absolute constant $c>0$, then with high probability~\Cref{eq:sequence} is a sequence of $\ell$ distinct points. (This is not necessarily the case if some of $x_i,y_j$ are allowed to be zero, but as we argued before, such queries are useless and we can assume without loss of generality that $x_i,y_j\neq 0$.) A version of this is formalized in~\Cref{lem:main_lemma} and~\Cref{lem:main_lemma_2}. Whenever this is the case, the sequence of outcomes in~\Cref{eq:sequence_h} consists of uniform and independent random bits -- indeed, $\bh$ is a uniformly random function and hence its evaluations on distinct points are independent and uniform. This happens for both the \yes and the \no distributions and as a result, any decision tree of depth $\ell\le 2^{cn}$ cannot sufficiently distinguish these distributions. We now sketch the proof that~\Cref{eq:sequence} is a distinct sequence with high probability. 

\paragraph*{Collision Probability Analysis.} 
This is done in~\Cref{sec:main_lemma}. We will analyze the probability that any two points in~\Cref{eq:sequence} are equal and show that it is at most $2^{-\Omega(n)}$. We then apply a union bound over all pairs of points (there are at most $\ell^2$ pairs) to conclude that with high probability, at least $1-\ell^2\cdot 2^{-\Omega(n)}\ge 2/3$, the sequence of points in~\Cref{eq:sequence} is distinct. This along with the above paragraph would complete the proof. 

Collisions within the $\bA_\dn x_i$  are significantly easier to analyze; the matrix $\bA$ is distributed according to a uniformly random invertible matrix, and hence intuitively the lower half $\bA_\dn$ has enough entropy to make collisions of the form $\bA_\dn x_i=\bA_\dn x'_i$ highly unlikely. Collisions within the $\bB_\up y_j$ can be similarly controlled. 
 We prove this in~\Cref{lem:main_lemma_2}. 
Collisions between $\bA_\dn x_i$ and $\bB_\up y_j$ are significantly harder to analyze, since $\bA_\dn$ and $\bB_\up$ are correlated with each other due to the relationship $\bB=(\bA^T)^{-1}$, but this can nevertheless be shown.

\paragraph*{The General Case: Forrelation $1$ versus $-1$.} We now describe the additional ideas required to prove the classical hardness of the \maxforrelationproblem. In order to generate instances with $\forr(f,g)=1$ and $\forr(f,g)=-1$, we are forced to sample $f$ to be a bent function and we are forced to set $g$ so that either $\chi_g=+\widehat{\chi_f}$ or $\chi_g=-\widehat{\chi_f}$. If we were to naively modify the aforementioned \yes distribution to a \no distribution by letting $\chi_g=-\widehat{\chi_f}$, then the resulting \yes and \no distributions would become easy to classically distinguish: querying the \yes distribution on $x=0$ and $y=0$ would produce identical answers, namely $(-1)^{\bh(0)}$ and $(-1)^{\bh(0)}$ while querying the \no distribution on these points would produce different answers, namely $(-1)^{\bh(0)}$ and $-(-1)^{\bh(0)}$. In order to get around this issue, in~\Cref{step2} of our actual hard distributions we need to apply an \emph{affine} transformation on the input variables, instead of a \emph{linear} transformation. This has the effect of shifting the origin, which intuitively means that the classical algorithm ``does not know'' which point to query in order to see this correlated pair of answers. Remarkably, applying this random shift also  simplifies the collision probability analysis in the case of collisions between $\bA_\dn x_i$ and $\bB_\up y_j$. For more details, see~\Cref{def:hard_distributions}.  

\subsubsection{Organization}
We describe our hard distributions in~\Cref{sec:hard-distributions}. We prove that these are indeed valid distributions (\Cref{lemma:hard_distributions}) and prove some some key properties about them (\Cref{corollary:marginal_B_1,corollary:marginal_B_1_uniform}). In~\Cref{sec:classical-lower-bound}, we present the main technical lemma (\Cref{lem:main_lemma}) and prove the main theorem assuming this in~\Cref{sec:main_theorem}. In~\Cref{sec:main_lemma}, we prove the main lemma.

%% file: sections/hard-distributions.tex

\section{Hard Distributions for the \maxforrelationproblem}
\label{sec:hard-distributions}

As described in the introduction, the hard instances of our problem will be based on the Maiorana-McFarland family (see e.g.~Section~6.1 of \cite{CM16}) of bent functions -- this family consists of ``inner-product-like'' functions. To sample our hard instances, we will first sample an affine shift, followed by a random ``inner-product-like'' function under this affine shift. We describe this in more detail below.

\subsection{Descriptions of Hard Distributions}

\paragraph{Notation.} For simplicity of notation, define the function $\chi:\bin\to \{\pm 1\}$ mapping $x$ to $(-1)^{x}$. For any Boolean-valued function $f:\bin^n\to \bin$, let $\chi_f:\bin^n\to \{\pm 1\}$ denote the function $\chi(f)(x)=(-1)^{f(x)}$. For matrices $A,B\in \bin^{n\times n}$, let $A=\begin{bmatrix}A_\up \\A_\dn \end{bmatrix}$ and $B=\begin{bmatrix}B_\up \\B_\dn \end{bmatrix}$ where $A_\up, A_\dn$ and $B_\up, B_\dn $ are $n/2\times n$ matrices representing the upper and lower halves of $A$ and $B$ respectively. We use $\langle x,y\rangle$ to denote the inner product (mod 2) between vectors $x,y\in \bin^n$. 

\paragraph*{}
We now proceed to the description of the hard distributions. We first define a joint distribution on 4-tuples $(A,B,a,b)$ where $A,B$ are matrices and $a,b$ are affine shifts (vectors) as follows. 

\begin{definition} \label{def:hard_matrices} Let $\bA\sim \bin^{n\times n}$ be a uniformly random matrix of full rank and let $\bB:=(\bA^T)^{-1}$.  Let $\ba\sim\bin^n$ be a uniformly vector and let $\bb=\bB^T\begin{bmatrix} \bB_\dn \\ \bB_\up\end{bmatrix}\ba$.
\end{definition} 

Let $\cL$ be the induced distribution on $(\bA,\bB)$. We use $\cL_{\dn\up}$,  $\cL_\dn$ and $\cL_\up$ to denote the induced distribution on 
$(\bA_\dn ,\bB_\up)$, $\bA_\dn$ and $\bB_\up$ respectively. We now define the hard instances of the \maxforrelationproblem ~for a family of functions $\cH$.  

\begin{definition}[Hard Instances of \maxforrelationproblem] \label{def:hard_distributions} 
Let $\cH$ be any collection of boolean functions mapping $\bin^{n/2}$ to $\{0,1\}$. Sample $(\bA,\bB,\ba,\bb)$ as in~\Cref{def:hard_matrices}. Sample $\bh:\bin^{n/2}\to \{0,1\}$ to be a uniformly random Boolean function in $\cH$. Define Boolean functions $\boldf,\bg:\bin^{n}\to\bin$ as follows.  
\begin{align*} 
\begin{split}
\boldf(x)&:= \abra{ \bA_\up x, \bA_\dn x}  + \abra{x,\ba} +\bh(\bA_\dn  x ) \\ \quad \bg(y)&:=\abra{ \bB_\up  y, \bB_\dn  y }  + \abra{y,\bb}+ \bh(\bB_\up  y+ \bB_\up \ba) +\abra{\bB_\up \ba,\bB_\dn \ba}.\end{split}
\end{align*}
Let $\muyes$ and $\muno$ be the induced distributions on $(\chi_{\boldf},\chi_{\bg})$ and $(\chi_{\boldf},-\chi_{\bg})$ respectively for $h\sim \cH$.
\end{definition}

We will later instantiate $\cH$ in various ways. The following lemma shows that regardless of $\cH$, these are indeed valid distributions, that is, $\muyes$ and $\muno$ are indeed supported on the \textsc{yes} and \textsc{no} instances of the \maxforrelationproblem.

\begin{lemma}\label{lemma:hard_distributions}
    For $(\chi_f,\chi_g)$ in the support of $\muyes$ and $\muno$, we have $\forr(f,g)=1$ and  $\forr(f,g)=-1$ respectively (regardless of $\cH$).
\end{lemma} 
\begin{proof}[Proof of~\Cref{lemma:hard_distributions}] 
Fix any $A,B,a,b$ as in~\Cref{def:hard_distributions} and let $f,g$ be as specified in \Cref{def:hard_distributions}. For any $y\in\bin^n$ , consider
\begin{align*}\widehat{\chi_f}(y)
&=\frac{1}{2^{n}} \sum_{z\in\bin^n} \chi_f(z)\cdot \minus{\abra{z,y}} \\
&= \frac{1}{2^{n}} \sum_{z\in\bin^n} \minus{f(z)+\abra{z,y}} 
\\&=\frac{1}{2^{n}} \sum_{z\in\bin^n}  \minus{\abra{ A_\up z ,A_\dn z }+\abra{z,a} +h(A_\dn z)+\abra{z,y}}.
\end{align*}
Since $A$ is invertible we can do a change of variables $x\leftarrow Az$, which in particular gives us  $x_\up\leftarrow A_\up z $ and $x_\dn\leftarrow A_\dn  z$. This lets us rewrite the above as 
\begin{align*}
\widehat{\chi_f}(y)=\frac{1}{2^{n}} \sum_{x\in\bin^n}  \minus{\abra{  x_\up,x_\dn} +\abra{A^{-1}x,a}+h(x_\dn)+\abra{ A^{-1}x,y}}.
\end{align*}
We now observe that $\abra{A^{-1}x,\circ}=\abra{x,B\circ}$, since $(A^{-1})^T=B$. Substituting this above, we see that
\begin{align*}
\widehat{\chi_f}(y)=\frac{1}{2^{n}} \sum_{x\in\bin^n}  \minus{\abra{  x_\up,x_\dn} +\abra{x,Ba}+h(x_\dn)+\abra{x,By}}.
\end{align*}
We now express $\abra{x,B\circ}$ as $\abra{x_\up,B_\up \circ}+\abra{x_\dn,B_\dn \circ}$. Thus,
\begin{align*}
\widehat{\chi_f}(y)=\frac{1}{2^{n}} \sum_{x\in\bin^n}  \minus{\abra{  x_\up,x_\dn} +\abra{x_\up,B_\up a}+\abra{x_\dn,B_\dn a}+h(x_\dn)+\abra{x_\up,B_\up y}+\abra{x_\dn,B_\dn y}}.
\end{align*}
We now group the terms based on $x_\up$. 
\begin{align*}
\widehat{\chi_f}(y)&=\frac{1}{2^{n}} \sum_{x\in\bin^n}  \minus{\abra{  x_\up,x_\dn+B_\up a +B_\up y} +\abra{x_\dn,B_\dn a+B_\dn y}+h(x_\dn)}\\
&=\frac{1}{2^{n}} \sum_{x_\dn\in \bin^{n/2}}\sum_{x_\up\in\bin^{n/2}}  \minus{\abra{  x_\up,x_\dn+B_\up a +B_\up y}}\cdot \minus{\abra{x_\dn,B_\dn a+B_\dn y}+h(x_\dn)}.
\end{align*}
Observe that the first term $\minus{\abra{  x_\up,x_\dn+B_\up a +B_\up y}}$ when summed over all $x_\up\in \bin^{n/2}$ is non-zero only if $x_\dn+B_\up a +B_\up y=0$ (equivalently, $x_\dn=B_\up a +B_\up y$); and if it is non-zero, it equals $2^{n/2}$. Thus, 
we have
\begin{align*}
\widehat{\chi_f}(y)&= \frac{1}{2^{n/2}}\minus{\abra{B_\up a +B_\up y,B_\dn a+B_\dn y}+h(B_\up a +B_\up y)}.
\end{align*}
We now expand the terms in the R.H.S. to obtain
\begin{align*}
   \widehat{\chi_f}(y)
    &=\frac{1}{2^{n/2}}\minus{\abra{B_\up y,B_\dn  y}+ \abra{\begin{bmatrix}B_\up  \\ B_\dn \end{bmatrix}y,\begin{bmatrix}B_\dn  \\ B_\up \end{bmatrix}a}+h(B_\up a +B_\up y)+ \abra{B_\up a,B_\dn a}}\\
     &= \frac{1}{2^{n/2}}  \minus{\abra{B_\up y,B_\dn  y}+ \abra{y,\begin{bmatrix}B_\up^T &  B_\dn^T \end{bmatrix}\begin{bmatrix}B_\dn  \\ B_\up \end{bmatrix}a}+h(B_\up a +B_\up y)+ \abra{B_\up a,B_\dn a}}.
\end{align*}
Recall that 
we have $b=B^T\begin{bmatrix}B_\dn  \\ B_\up \end{bmatrix}a=\begin{bmatrix}B_\up^T &  B_\dn^T \end{bmatrix}\begin{bmatrix}B_\dn  \\ B_\up \end{bmatrix}a$. Substituting this in the above equation, we have  
\begin{align*}\widehat{\chi_f}(y)&=\frac{1}{2^{n/2}} \minus{\abra{B_\up y,B_\dn y}+\abra{y,b}+h(B_\up a +B_\up y)}+\abra{B_\up a,B_\dn a}.
\end{align*}
Recalling the definition of $g(y)$ from~\Cref{def:hard_distributions}, we see that
\begin{align}
\widehat{\chi_f}(y)&\triangleq \frac{1}{2^{n/2}}\chi_g(y).\label{eq:proof_fourier_transform}
\end{align}
We now use the defining equation for the Forrelation function (\Cref{eq:def_forr}) to get \[\forr(\chi_f,\chi_g)\triangleq \frac{1}{2^{n/2}}\sum_{y\in\bin^n}\widehat{\chi_f}(y)\chi_g(y).\] 
Combining this and~\Cref{eq:proof_fourier_transform}, we see that $\forr(\chi_f,\chi_g)=\frac{1}{2^{n}}\sum_y \chi_g(y)^2=1$. It can be similarly shown that $\forr(\chi_f,-\chi_g)=\frac{1}{2^{n}}\sum_y (-1)\cdot \chi_g(y)^2=-1$. This completes the proof of~\Cref{lemma:hard_distributions}.
\end{proof}

\subsection{Characterizing the Marginals of the Distributions}
Recall that we used $\cL$ to denote the induced distribution on $(\bA,\bB)$, and $\cL_\dn$ and $\cL_\up$ to denote the induced distribution on 
$\bA_\dn$ and $\bB_\up$ respectively. In this section, we will characterize the marginal distributions $\cL_\up$ and $\cL_\dn$, which turn out to be identical, as follows:

\begin{lemma}
\label{corollary:marginal_B_1}
Each of the distributions $\cL_{\up}$ and $\cL_\dn$ is precisely the uniform distribution over all matrices in $\bin^{n/2\times n}$  that have full row-rank. \end{lemma}

We then show the following fact about the row-rank of a uniformly random rectangular matrix.\begin{fact}\label{fact:full_row_rank} A uniformly random matrix in $\bin^{n/2\times n}$ has full row-rank with probability at least $1-(n/2)2^{-n/2}$. \end{fact}

As an immediate corollary of this and~\Cref{corollary:marginal_B_1}, we obtain the following.
\begin{corollary}
\label{corollary:marginal_B_1_uniform}
Each of the distributions $\cL_{\up}$ and $\cL_{\dn}$ is $(n/2)\cdot 2^{-n/2}-$close to the uniform distribution over $\bin^{n/2\times n}$ in total variational distance. \end{corollary}

We now give the proofs of~\Cref{corollary:marginal_B_1} and~\Cref{fact:full_row_rank}.

\begin{proof}[Proof of~\Cref{corollary:marginal_B_1}]
We will prove this for the distribution $\cL_\up$ and the argument for $\cL_\dn$ is identical. Recall that $\bB$ is sampled according to a uniformly random full-rank matrix. Fix any full row-rank matrix $B_\up$. We will count the number of matrices $B_\dn$ such $(B_\up,B_\dn)$ has full rank. We will show that this number is precisely $(2^{n}-2^{n/2})\times (2^{n}-2^{n/2+1})\times \ldots\times (2^{n/2}-2^{n-1})$. This would complete the proof.

To count the number of $B_\dn$, we first count the number of possibilities for each row of $B_\dn$. Firstly, each row of $B_\dn$ must not lie in the span of the previous rows of $B_\dn$ and the rows of $B_\up$. The first row of $B_\dn$ is any vector not in the span of the rows of $B_\up$ and thus has $2^{n}-2^{n/2}$ possibilities. Having fixed this, the second row of $B_\dn$ can be any vector that is not in the span of the first row of $B_\dn$ and the rows of $B_\up$, and thus has $2^{n}-2^{n/2+1}$ possibilities. We repeat this argument and at the $i$-th step, we choose a vector that is not in the $n/2+i-1$-dimensional space spanned by the first $i-1$ rows of $B_\dn$ and the $n/2$ rows of $B_\up$; this can be done in $2^{n}-2^{n/2+i-1}$ ways. Doing this for all $n/2$ rows shows that the total number of possibilities for $B_\dn$ is precisely $\prod_{i=1}^{n/2}(2^{n}-2^{n/2+i-1})$ and this completes the proof of~\Cref{corollary:marginal_B_1}.
\end{proof}

\begin{proof}[Proof of \cref{fact:full_row_rank}]
We perform a calculation identical to that in~\Cref{corollary:marginal_B_1}. By a similar argument, we can show that the number of matrices in $\bin^{n/2 \times n}$ with full row-rank is precisely 
\[\prod_{i=1}^{n/2} (2^n-2^{i-1}).\]
Thus, the probability that a uniformly random $n/2\times n$ matrix is of full row-rank is precisely
\begin{align*}
\frac{\prod_{i=1}^{n/2}(2^n-2^{i-1})}{2^{n^2/2}} =\prod_{i=1}^{n/2}\pbra{\frac{2^n-2^{i-1}}{2^n}} \ge (1-2^{-n/2})^{n/2}\ge 1 - (n/2)2^{-n/2}.\end{align*}
This completes the proof of~\Cref{fact:full_row_rank}.
\end{proof}

%% file: sections/classical-lower-bound.tex

\section{Classical Lower Bound} \label{sec:classical-lower-bound}

The main technical ingredient in the classical lower bound will be~\Cref{lem:main_lemma}. We will describe this lemma and its proof in~\Cref{sec:main_lemma}. We will then prove~\Cref{thm:main_theorem} in~\Cref{sec:main_theorem} assuming this. 

\subsection{Statement of \texorpdfstring{\Cref{lem:main_lemma} and its Proof}{Main Lemma}}
\label{sec:main_lemma}

The main technical ingredient is the following.
\begin{restatable}{lem}{mainlemma}
\label{lem:main_lemma}
Let $x,y\in \{0,1\}^n$ be any two vectors. Then, 
\[ \Prx_{\substack{(\bA_\dn ,\bB_\up )\sim \cL_{\dn\up}\\ \ba\sim\bin^n}}[\bA_\dn x=\bB_\up  y + \bB_\up \ba] = 2^{-n/2}.\]
\end{restatable}
In the above lemma, $(\bA_\dn ,\bB_\up )\sim \cL_{\dn\up}$ and $\ba\sim \bin^n$ are sampled independently, just as in~\Cref{def:hard_matrices}. We will also require the following lemma. 
\begin{restatable}{lem}{mainlemmatwo}
\label{lem:main_lemma_2}
For $x\neq x'\in \bin^n$ and $y\neq y'\in \bin^n$, we have 
\[\Prx_{\bA_\dn\sim \cL_\dn}[\bA_\dn x=\bA_\dn x']\le (n/2+1)\cdot 2^{-n/2}\]
\[\Prx_{\bB_\up\sim \cL_\up}[\bB_\up y=\bB_\up y']\le (n/2+1)\cdot 2^{-n/2}.\]
\end{restatable}

 We now prove these two lemmas.

 \begin{proof}[Proof of~\Cref{lem:main_lemma}]
Let $\cE$ be the event $\bA_\dn x = \bB_\up y + \bB_\up \ba$. This is equivalent to $\bB_\up \ba=\bA_\dn x+  \bB_\up y$. Now, regardless of what $x,y$ are, the vector $\ba$ is distributed as a uniformly random vector in $\bin^n$ that is \emph{independent} of $\bA,\bB$, because $\ba$ is sampled uniformly and independently of $\bA,\bB$. Furthermore, $\bB_\up$ has full row-rank. Therefore, fixing $\bA=A$ and $\bB=B$, we see that the probability over $\ba$ that $B_\up\ba$ is equal to $A_\dn x + B_\up y$ is precisely $2^{-n/2}$. This completes the proof.
 \end{proof}

 \begin{proof}[Proof of~\Cref{lem:main_lemma_2}]We prove this lemma for $\bA_\dn\sim \cL_\dn$ and the proof of $\bB_\up\sim \cL_\up$ is identical. The event we wish to bound the probability of is $[\bA_\dn(x-x')=0]$.  We use~\Cref{corollary:marginal_B_1_uniform} to conclude that the total variational distance between $\cL_\dn$ and the uniform distribution is at most $(n/2)\cdot 2^{-n/2}$. Let us now work with a uniformly random matrix $\bA_\dn$. Since $(x-x')\neq 0$, the vector $\bA_\dn(x-x')$ is a uniformly random vector in $\bin^{n/2}$. Therefore, the probability that it is zero is at most $2^{-n/2}$. This completes the proof. \end{proof}

 We now complete the proof of~\Cref{thm:main_theorem} using these lemmas.

 \subsection{\texorpdfstring{Proof of~\Cref{thm:main_theorem}}{Proof of Main Theorem}}
\label{sec:main_theorem}

Let $\cH$ be any family of all $n/2$-variate boolean functions. Consider a classical randomized query protocol for the \maxforrelationproblem ~with $D$ queries. Recall from~\Cref{lemma:hard_distributions} that $\muyes$ and $\muno$ are supported on the \textsc{yes} and \textsc{no} instances of the \maxforrelationproblem. Given a randomized query protocol with $D$ queries for the \maxforrelationproblem  ~that succeeds with at least $2/3$ probability, by Yao's principle, there exists a \emph{deterministic} decision tree of depth $D$ that distinguishes $\muyes$ and $\muno$ with advantage at least $1/3$. 

Given such a deterministic decision tree of depth $D\le 2^{n/4}/6$ and a root-to-leaf path $\cP$ of length $\ell$ in the tree, each node in the path $\cP$ corresponds to either a query of the form $f(x)$ (where the truth table of $f$ is probed), or a query of the form $g(y)$ (where the truth table of $g$ is probed). We assume without loss of generality that the query vectors for $f$ are distinct and similarly the query vectors for $g$ are distinct (there may be common query vectors which are given to both $f$ and $g$). We will then use the following claims.

\begin{claim}\label{claim:collision_probability} 
Consider any deterministic decision tree of depth $D\le 2^{n/4}/(6 \sqrt{n})$, and fix any root-to-leaf path $\cP$ of length $\ell \leq D$ in the tree.
Let $x^{(1)},\ldots,x^{(k)}\in\bin^n$ and $y^{(1)},\ldots,y^{(\ell-k)}\in\bin^n$ be the sequence of vectors queried. 
Let $\bA,\bB,\ba$ be distributed as in \Cref{def:hard_matrices}, and consider the sequence of points
\begin{equation} \label{eq:key-sequence}
\bA_\dn x^{(1)},\ldots,\bA_\dn x^{(k)} \quad \text{and}\quad \bB_\up y^{(1)} +\bB_\up \ba,\ldots,\bB_\up y^{(\ell- k)}+\bB_\up \ba.
\end{equation}
Let $\cE$ be the event that \eqref{eq:key-sequence} is a sequence of $\ell$ \emph{distinct} points. Then, $\Pr_{\muyes}[\cE],\Pr_{\muno}[\cE]\ge 9/10$. 
\end{claim}

\begin{claim}\label{claim:path_probability} Let $\cH$ be the family of all $n/2$-variate boolean functions.
Under the same hypothesis as~\Cref{claim:collision_probability}, whenever $\cE$ occurs, the probability of taking path $\cP$ for the distributions $\muyes|\cE$ and $\muno|\cE$ is exactly $2^{-\ell}$.\end{claim}
%
%
%
%

\paragraph*{Remark.} \Cref{claim:path_probability} is the \emph{only} place where the properties of $\cH$ come into play -- every other part of the proof is \emph{independent} of $\cH$.

Once we have these claims, the proof follows quite easily. Let $\cH$ be the family of all $n/2$-variate boolean functions. Let us look at the induced distributions $\muyesl$ and $\munol$ on the leaves of the decision tree when the inputs are sampled according to $\muyes$ and $\muno$ respectively and let $\mu$ be the distribution on the leaves induced by a truly random walk down the tree.~\Cref{claim:collision_probability} and ~\Cref{claim:path_probability} imply that for any leaf, the probability that $\muyesl$ and $\munol$ assign to that leaf are each at least $9/10$ times the probability assigned by $\mu$. This implies that  there exists distributions ${\muyesltil},{\munoltil}$ on the leaves such that
\[ \muyesl = \tfrac{9}{10} \mu+ \tfrac{1}{10}{\muyesltil},\]
\[ \munol = \tfrac{9}{10} \mu+ \tfrac{1}{10}{\munoltil}.\]
This implies that the total variational distance between $\muyesl$ and $\munol$ is at most $1/10$ and hence, the output distributions of the decision tree on inputs sampled according to $\muyes$ and $\muno$ differ in total variational distance by at most $1/10$. This contradicts the assumption that the decision tree distinguishes these distributions with advantage at least $1/3$. This completes the proof of~\Cref{thm:main_theorem} assuming~\Cref{claim:collision_probability} and~\Cref{claim:path_probability}. We will now prove these claims.

\subsection{Proof of \texorpdfstring{\Cref{claim:collision_probability}}{Claim 3.1} and \texorpdfstring{\Cref{claim:path_probability}}{Claim 3.2}}

We now proceed to the proof of~\Cref{claim:collision_probability}, which is where we will use \Cref{lem:main_lemma} and \Cref{lem:main_lemma_2}. We will then prove~\Cref{claim:path_probability} which relies primarily on the properties of $\cH$.

\begin{proof}[Proof of~\Cref{claim:collision_probability}]
 We first analyze the probability of $\cE$ when the distribution on inputs is $\muyes$. The calculation for $\muno$ is identical and is omitted. Note that \eqref{eq:key-sequence} is a random sequence of points (where the randomness comes from $\bA,\bB$ and $\ba$). As stated in \Cref{claim:path_probability}, let $\cE$ be the event that this is a sequence of $\ell$ \emph{distinct} points. We will now argue that $\cE$ is a high-probability event.  

First, let us bound the probability of collisions within the $\bA_\dn x^{(i)}$. Fix any $i\neq i'\in [k]$. We apply~\Cref{lem:main_lemma_2} to the vectors $x^{(i)}\neq x^{(i')}$. ~\Cref{lem:main_lemma_2} implies that
\[ \Prx[\bA_\up x^{(i)}=\bA_\up x^{(i')}]\le (n/2+1)\cdot 2^{-n/2}. \]
We now apply a union bound over $i,i'\in[k]$. There are at most $k^2$ possibilities to union bound over. This implies that with probability at least $1-k^2\cdot(n/2+1)\cdot 2^{-n/2}$, the sequence of points 
\[\bA_\dn x^{(1)},\ldots,\bA_\dn x^{(k)}\] 
is a sequence of $k$ distinct points. We can similarly argue about collisions within the $\bB_\up y^{(j)}$ to conclude that with probability at least $1-(\ell-k)^2\cdot (n/2+1)\cdot 2^{-n/2}$, the sequence of points 
\[\bB_\up y^{(1)} + \bB_\up \ba ,\ldots,\bB_\up y^{(\ell-k)} + \bB_\up \ba\] 
is a sequence of $\ell-k$ distinct points. Finally, we argue about collisions between pairs of the form $\bA_\dn x^{(i)}$ and $\bB_\up y^{(j)} +\bB_\up \ba$. Fix any $i\in[k]$ and $j\in[\ell-k]$. We apply~\Cref{lem:main_lemma} to the vectors $x^{(i)}$ and $y^{(j)}$. ~\Cref{lem:main_lemma} implies that
\[ \Prx[\bA_\dn  x^{(i)}=\bB_\up y^{(j)}+\bB_\up \ba]\le   2^{-n/2}.\]
We now apply a union bound over $(i,j)$. There are at most $k\cdot (\ell-k)$ possibilities to union bound over. This implies that with probability at least $1- k\cdot (\ell-k) \cdot 2^{-n/2}$, 
there are no collisions among pairs of the form $\bA_\dn x^{(i)}$ and $\bB_\up y^{(j)}+\bB_\up \ba$. So by a union bound, the total probability of the bad event $\neg \cE$ is at most
\[k^2 \cdot (n/2+1)\cdot 2^{-n/2} + (\ell-k)^2\cdot (n/2+1)\cdot 2^{-n/2} + k\cdot (\ell-k)\cdot 4\cdot 2^{-n/2}\le 3\cdot D^2\cdot n\cdot 2^{-n/2}.\]
for large enough $n$. Recall that we set $D\le 2^{n/4}/(6\sqrt{n})$, so we get that
\[ 3\cdot D^2\cdot n\cdot 2^{-n/2}\le 3\cdot 2^{n/2}\cdot\tfrac{1}{36n} \cdot n \cdot 2^{-n/2} < \tfrac{1}{10}.\]
This shows that $\Prx_{\muyes}[\cE]\ge 9/10$. The argument for $\muno$ is identical. This completes the proof of~\Cref{claim:collision_probability}.
\end{proof}

\begin{proof}[Proof of~\Cref{claim:path_probability}]
 We first analyze the probability of receiving any particular sequence of outcomes when the distribution on inputs is $\muyes$. The calculation for $\muno$ is identical and is omitted.

Recall the $\muyes$ distribution. This is obtained by sampling $\bA,\bB,\ba,\bb$ as in~\Cref{def:hard_distributions}, sampling $\bh$ uniformly at random from $\cH$, and outputting $(\chi_{\boldf},\chi_{\bg})$ where $\boldf,\bg:\bin^n\to \bin$ are defined as:
\begin{align*} 
\begin{split}
\boldf(x)&:= \abra{ \bA_\up x, \bA_\dn x}  + \abra{x,\ba} +\bh(\bA_\dn  x ) \\ \quad \bg(y)&:=\abra{ \bB_\up  y, \bB_\dn  y }  + \abra{y,\bb}+ \bh(\bB_\up  y+ \bB_\up \ba) +\abra{\bB_\up \ba,\bB_\dn \ba}.\end{split}
\end{align*}
Along the path $\cP$, we have queried $\boldf(x^{(1)}),\ldots,\boldf(x^{(k)})$ and $\bg(y^{(1)}),\ldots,\bg(y^{(\ell-k)})$. Let us consider the contribution of $\bh(\circ)$ to these query responses. This is given by evaluating $\bh(\circ)$ on the following sequence of points
\[\bA_\dn x^{(1)},\ldots,\bA_\dn x^{(k)} \quad \text{and}\quad \bB_\up y^{(1)} +\bB_\up \ba,\ldots,\bB_\up y^{(\ell- k)}+\bB_\up \ba\]
which is precisely the sequence given in \Cref{eq:key-sequence}. We will now argue that when $\cE$ happens, the probability of taking this path under $\muyes$ (and similarly $\muno$) is precisely $2^{-\ell}$. Let us compute the probability of taking the path $\cP$ under $\muyes$ conditioned on $\cE$ happening. When $\cE$ happens, we have that the sequence of points
\[\bA_\dn x^{(1)},\ldots,\bA_\dn x^{(k)} \quad \text{and}\quad \bB_\up y^{(1)} +\bB_\up \ba,\ldots,\bB_\up y^{(\ell- k)}+\bB_\up \ba\] 
is a sequence of $\ell$ distinct points. We now observe that the evaluations of the function $\bh$ on these points are independent and uniformly random bits in $\{0,1\}$. In other words,
\[\bh(\bA_\dn x^{(1)}),\ldots,\bh(\bA_\dn x^{(k)}) \quad \text{and}\quad \bh(\bB_\up y^{(1)} +\bB_\up \ba),\ldots,\bh(\bB_\up y^{(\ell- k)}+\bB_\up \ba)\] 
is a sequence of uniformly random bits in $\{0,1\}$ when $\bh\sim \cH$. Thus, the probability of receiving any particular sequence of outcomes when querying the truth tables of $\boldf,\bg$ at the $\ell$ vectors $x^{(1)},\ldots,x^{(k)},$ $y^{(1)},\ldots,y^{(\ell-k)}$ is either exactly $2^{-\ell}$. This completes the proof.

\end{proof}

\subsection{Proof of \texorpdfstring{\Cref{thm:main_corollary}}{Corollary 1.6}}

 \begin{proof}[Proof of~\Cref{thm:main_corollary}]
 Let $\cH$ be the family of all boolean functions $\{h:\{0,1\}^{n/2}\to \{0,1\}\}$. It is well known~\cite{HILL99,GGM86} that if one-way functions exist, then we can construct a family of pseudorandom functions $\cH':=\{h_\lambda:\{0,1\}^{n/2}\to \{0,1\}\}_{\lambda\in\{0,1\}^{k(n)}}$ with $k(n)=\poly(n)$ such that
 \begin{itemize}
 \item \textbf{Efficiency:} there is a $\poly(n)$-sized classical circuit that computes $h_\lambda(x)$ given inputs $\lambda\in\{0,1\}^{k(n)}$ and $x\in \{0,1\}^{n/2}$.
 \item \textbf{Security:} Any classical polynomial-time algorithm $\cA$ that queries the truth-table of an $n/2$-bit Boolean function cannot sufficiently distinguish a uniformly random function in $\cH'$ from a truly uniformly random function, i.e.,
 \[ \abs{ \Ex_{\blambda\sim\{0,1\}^{k(n)}}\sbra{\cA^{h_{\blambda(x)}}(1^n)} - \Ex_{\bh\sim \cH}\sbra{\cA^{\bh(x)}(1^n)} } \le \negl(n). \]
 \end{itemize}

 Let $\muyes$ and $\muno$ be the hard distributions as in~\Cref{def:hard_distributions} defined with respect to $\cH$ and similarly $\muyesprime$ and $\munoprime$ be defined with respect to $\cH'$. 
 
 \paragraph*{Implementability of $f,g$.} For any pair $(f,g)$ drawn from either $\muyesprime$ or $\munoprime$, the functions $f,g$ are computable by polynomial-sized circuits. In more detail, for a fixed draw of $A,B,a,b,\lambda$, the circuit for $f$ takes input $x$,  first computes the inner products $\abra{ A_\up x,A_\dn x}$ and $\abra{x,a}$, then computes the vector $A_\dn x$, and finally applies the circuit for $h_\lambda$ on this vector and thus computes $f(x)=\abra{A_\up x,A_\dn x}+\abra{x,a}+h_\lambda(A_\dn x)$. All of these operations can be implemented by $\poly(n)$-sized classical circuits and the circuit for $g$ is analogous. We will now show that under the cryptographic assumption, there is no classical algorithm that distinguishes $\muyesprime$ and $\munoprime$ 
with $\poly(n)$ queries. 

\paragraph*{Classical Indistinguishability of $f,g$.} Let $\cA'$ be any classical algorithm that makes at most $\poly(n)$ queries to the truth tables of $f$ and $g$. \Cref{thm:main_theorem} shows that $\cA'$ cannot distinguish $\muyes$ and $\muno$ with advantage more than $\negl(n)$\footnote{While the statement of~\Cref{thm:main_theorem} only refers to $1/3$ advantage, the proof (\Cref{claim:collision_probability}) shows that the advantage of a classical algorithm in solving the \maxforrelationproblem ~scales as $2^{-n/2}$ times $D^2$ where $D$ is the number of queries. In particular, for $\poly(n)$-time algorithms (which make at most at most $\poly(n)$ queries), this advantage is $\negl(n)$.}. We will now show that under the cryptographic assumption, $\cA'$ cannot distinguish $\muyes$ and $\muyesprime$ with more than $\negl(n)$ advantage. By the same argument, an analogous statement holds for the distributions $\muno$ and $\munoprime$. Consequently, by the triangle inequality, we get that $\cA'$ cannot distinguish $\muyesprime$ and $\munoprime$ with  more than $\negl(n)$ advantage and this completes the proof.

To see that $\cA'$ cannot distinguish $\muyes$ and $\muyesprime$ with more than $\negl(n)$ advantage, we show that any such algorithm can be turned into a distinguisher for $\cH$ and $\cH'$. Indeed, consider the algorithm $\cA$ that given query access to an unknown $n/2$-bit boolean function $h$, first samples $(\bA,\bB,\ba,\bb)$ as in~\Cref{def:hard_matrices} and runs $\cA'$ on the pair of functions $(f,g)$ as defined in the \yes distribution in ~\Cref{def:hard_distributions}. Similarly to the argument before, each query to $f$ at the point $x$ can be simulated by making a query to $h$ at the point $A_\dn x$ and computing $f(x)=\abra{A_\up x,A_\dn x}+\abra{x,a}+h(A_\dn x)$, and similarly for $g$. Thus, a distinguisher for $\muyesprime$ and $\muyes$ can be turned into one for $\cH'$ and $\cH$ with the same number of queries and same advantage. This completes the proof.
 \end{proof}

%% file: sections/appendix_2.tex
\section{Appendix: Upper bounding the maximum Rorrelation for a Haar random orthogonal matrix}

In this appendix we establish \Cref{claim:rorr} from the Introduction.
Recall the definition of the \emph{Rorrelation} function, 
\[
\rorr_U(f,g):=2^{-n} \sum_{x,y\in\bin^n}f(x)g(y)U_{x,y},
\] 
where $U$ is any $2^n\times 2^n$ unitary matrix. Let $N=2^n$.
 
\begin{lemma}
\label{lem:max_rorr} Consider the Rorrelation function $\rorr_{\bU}$ for an $N\times N$ Haar random orthogonal matrix $\bU$. With probability at least $1-2^{-\Theta(N)}$, we have
\[ \max_{f,g:\bin^n\to \{\pm 1\}} \rorr_{\bU}(f,g)\le 0.99.\]
\end{lemma} 
A similar proof can be used for the event $\min_{f,g:\bin^n\to \{\pm 1\}}\rorr_{\bU}(f,g)\ge -0.99$, and this establishes that the largest value of $|\rorr_{\bU}(f,g)|$ is at most $0.99$ with probability at least $1-2^{-\Theta(N)}$. We now prove~\Cref{lem:max_rorr}.
\begin{proof}[Proof of~\Cref{lem:max_rorr}] Fix any Boolean function $f:\bin^n\to \{\pm 1\}$ and let us optimize over $g:\bin^n\to\{\pm 1\}$. Recall that (for a fixed given $U$)
\begin{equation}
\label{eq:max2}
\max_{g:\bin^n\to \{\pm 1\}} \rorr_U(f,g)\triangleq \max_{g:\bin^n\to \{\pm 1\}} 2^{-n}\sum_{y}\pbra{\sum_x U_{x,y}f(x)}g(y)\end{equation} 
It is easy to see that the optimal solution $g^*$ to~\Cref{eq:max2} over all $\{\pm 1\}$-valued functions is $g^*(y):=\sign(\sum_x U_{x,y} f(x))$, since this choice of $g^*$ makes all the summands in the outer sum non-negative. Interpreting $f$ as a vector in $\{\pm 1\}^N$ denoting the truth table of the function $f$, the optimizer to~\Cref{eq:max2} is $g^*=\sign(Uf)$ and the optimal value is  $N^{-1}\|Uf\|_1$. Thus,  
\begin{equation} \label{eq:max3} \max_{f,g:\bin^n\to \{\pm 1\}} \rorr_U(f,g) = N^{-1} \max_{f:\bin^n\to\{\pm 1\}} \|U f\|_1.\end{equation} 
We will now argue that for any fixed $f\in \{\pm 1\}^N$, with overwhelming probability (at least $1-3^{-N})$ over a Haar random orthogonal matrix $\bU$, we have $\|\bU f\|_1\le 0.99  N$. This, along with a union bound over all $f\in \{\pm 1\}^N$ (there are at most $2^{N}$ such functions), implies that with probability at least $1-(3/2)^N$, for all $\pm 1$-valued functions $f$, we have $\|\bU f\|_1\le 0.99 N$. This, along with~\Cref{eq:max3}, completes the proof of~\Cref{lem:max_rorr}.

Fix any $f\in \{\pm 1\}^N$. Observe that $\|f\|_2=\sqrt{N}$ and hence $v:=\tfrac{1}{\sqrt{N}}f$ is a unit vector in $\R^N$. Hence as we vary over Haar random orthogonal matrices $\bU$, the vector $\bU v$ is a Haar random unit vector in $\Su^{N-1}$. Intuitively, such a vector is highly unlikely to have $\ell_1$ norm larger than $0.99\sqrt{N}$ (recall that by Cauchy-Schwarz, $\sqrt{N}$ is the largest possible $\ell_1$ norm of any vector in $\Su^{N-1}$). More precisely,~\Cref{fact:spherical_cap} implies that with probability least $1-3^{-N}$ over a Haar random orthogonal matrix $\bU$, we have
\[ \|\bU v\|_1 \le 0.99 \sqrt{N}.\]
This completes the proof. 
\end{proof}

\begin{fact} \label{fact:spherical_cap} For a Haar random unit vector $\bu\sim \Su^{N-1}$, we have
\[ \Pr[\|\bu\|_1\ge 0.99\sqrt{N}]\le 3^{-N}.\]
\end{fact}
\begin{proof} We use $\vol,\area$ to denote the Lebesgue volume and surface area of subsets of $\R^N$ and $\mu$ to denote the Haar measure on $\Su^{N-1}$. Observe that for any measurable subset $T\subseteq \Su^{N-1}$, we have 
\begin{equation}
    \label{eq:app_1}
    \mu(T) \triangleq  \frac{\area(T)}{\area(\Su^{N-1})}.
\end{equation}

Consider the intersection of the unit sphere $\Su^{N-1}$ (depicted in black in~\Cref{fig:circle}) and the complement of the $\ell_1$-ball, $S:=\{ v\in \R^N:\|v\|_1\ge 0.99 \sqrt{N}\}$; our goal is to bound $\mu(\Su^{N-1}\cap S)$, the measure of the intersection under the Haar measure on $\Su^{N-1}$. We can view the set $S$ as a union of $2^N$ many half spaces, namely, for each $\alpha\in \{\pm 1\}^N$, we have the half-space $S_\alpha:=\{v\in \R^N:\sum_{i=1}^N \alpha_i v_i \ge 0.99\sqrt{N}\}$. For each $\alpha\in \{\pm 1\}^N$, the set $\Su^{N-1}\cap S_\alpha$ is a spherical cap, and we will show that the measure of this cap is at most $7^{-(N-1)}$. A union bound over all $2^N$ choices of $\alpha\in\{\pm 1\}^N$ gives that with probability at least $1-2^N\cdot 7^{-(N-1)}\ge 1 - 3^{-N}$ (for large $N$), we have $\|u\|_1\le 0.99\sqrt{N}$, and this completes the proof.

We now upper bound the measure of the spherical cap $\Su^{N-1}\cap S_\alpha$ (depicted in red in~\Cref{fig:circle}). Observe that the diameter (distance between two furthest points) of this spherical cap is $d:=2 \sqrt{1-(0.99)^2}$ and that $d\le 2/7$. Consider the sphere $B'$ in $\R^N$ of diameter $d$ that intersects the spherical cap precisely at its edge (depicted in blue in~\Cref{fig:circle}); note that there is a unique such sphere $B'$ and that the intersection of $B'$ with the spherical cap $\Su^{N-1}\cap S_\alpha$ forms a great circle of $B'$. Now, consider the intersection of the interiors of $B'$ and $\Su^{N-1}$ (depicted in yellow in~\Cref{fig:circle}).  This is a convex body that is contained in $B'$ and hence its surface area is at most that of $B'$ (this inequality is a known direct consequence of Cauchy's surface area formula, see e.g. \cite{overflow-area} and Section~5.5 of \cite{KlainRota}).
Since the boundary of $\Su^{N-1}\cap S_\alpha$ is a subset of the boundary of this convex body, the surface area of $\Su^{N-1} \cap S_\alpha$ is at most that of $B'$.  Therefore, by~\Cref{eq:app_1}, we have the following inequality
\[\mu(\Su^{N-1} \cap S_\alpha)\triangleq \frac{\area(\Su^{N-1} \cap S_\alpha)}{\area(\Su^{N-1})}\le \frac{\area(B')}{\area(\Su^{N-1})}=(d/2)^{N-1}\le 7^{-(N-1)}, \]
where the second-to-last inequality uses the fact that a ball of radius $d/2$ in $\R^N$ has surface area $(d/2)^{N-1}$ times that of the unit ball in $\R^N$. This completes the proof. 
\end{proof}

\begin{figure}
\centering

\tikzset{every picture/.style={line width=0.75pt}} 

\begin{tikzpicture}[x=0.75pt,y=0.75pt,yscale=-1,xscale=1]

\draw [draw opacity=0][fill={rgb, 255:red, 248; green, 231; blue, 28 }  ,fill opacity=0.57 ]   (297.12,121.79) .. controls (275.9,135.37) and (252.33,132.67) .. (238.94,119.62) .. controls (225.54,106.57) and (224.8,64.47) .. (241.11,61.44) .. controls (257.43,58.41) and (294.33,88.67) .. (298.12,121.79) ;
\draw   (163.53,189.51) .. controls (134.05,157.74) and (135.91,108.09) .. (167.68,78.61) .. controls (199.44,49.13) and (249.09,50.99) .. (278.57,82.75) .. controls (308.05,114.52) and (306.2,164.17) .. (274.43,193.65) .. controls (242.66,223.13) and (193.01,221.27) .. (163.53,189.51) -- cycle ;
\draw  [color={rgb, 255:red, 30; green, 87; blue, 244 }  ,draw opacity=1 ][line width=1.5]  (238.94,118.62) .. controls (223.47,101.95) and (224.45,75.9) .. (241.11,60.44) .. controls (257.78,44.97) and (283.82,45.95) .. (299.29,62.61) .. controls (314.75,79.28) and (313.78,105.33) .. (297.12,120.79) .. controls (280.45,136.26) and (254.4,135.28) .. (238.94,118.62) -- cycle ;
\draw    (221.05,135.13) -- (269.43,89.88) ;
\draw  [dash pattern={on 4.5pt off 4.5pt}]  (191.2,7.39) -- (347.02,173.84) ;
\draw    (221.05,135.13) -- (241.11,60.44) ;
\draw [color={rgb, 255:red, 255; green, 0; blue, 0 }  ,draw opacity=1 ][line width=1.5]    (241.11,60.44) .. controls (276,70) and (290,92) .. (298.12,120.79) ;

\draw (212.36,105.93) node [anchor=north west][inner sep=0.75pt]  [font=\footnotesize,rotate=-0.81]  {$1$};
\draw (339,133.4) node [anchor=north west][inner sep=0.75pt]  [font=\small]  {$\sum_{i=1}^N \alpha_iv_i\ge 0.99\sqrt{N}$};
\draw (253.36,106.93) node [anchor=north west][inner sep=0.75pt]  [font=\footnotesize,rotate=-0.81]  {$0.99$};

\end{tikzpicture}

\caption{Proof of~\Cref{fact:spherical_cap}. Depiction of $\Su^{N-1}$ in black, $\Su^{N-1}\cap S_\alpha$ in red, $B'$ in blue.}
     \label{fig:circle}
\end{figure}
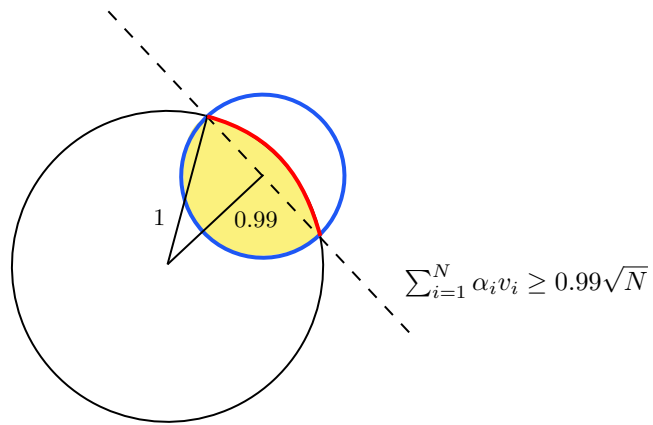